\documentclass[11pt, oneside]{amsart}   	% use "amsart" instead of "article" for AMSLaTeX format
\usepackage{geometry}                		% See geometry.pdf to learn the layout options. There are lots.
\usepackage{graphicx}				% Use pdf, png, jpg, or eps§ with pdflatex; use eps in DVI mode
\usepackage{cancel}					% TeX will automatically convert eps --> pdf in pdflatex		

\usepackage{amssymb}

\DeclareMathOperator{\Ha}{\mathcal{H}}
\DeclareMathOperator{\Id}{id}
\DeclareMathOperator{\diag}{diag}

\newcommand{\R}{\mathbb{R}}

\newcommand{\p}{\partial}

\newcommand{\phii}{(u ^i\vert u^*_{-i})}
\newtheorem{theorem}{Theorem}
\newtheorem{proposition}{Proposition}

\newtheorem{remark}{Remark}

\newtheorem{definition}{Definition}
\newtheorem{lemma}{Lemma}

\begin{document}

\title[What does a dynamic oligopoly maximize? The Markov case]{What does a dynamic oligopoly maximize? \\
The continuous time Markov case%\\ with consumption externalities
}
\author{%Hern\'andez--Lerma \& 
Juan Pablo Rinc\'on--Zapatero}
\address{Departamento de Econom\'\i a, Universidad Carlos III de Madrid}
\email{jrincon@eco.uc3m.es}
\urladdr{https://www.eco.uc3m.es/$\sim$jrincon/index.html}
\thanks{Support from the Spanish Ministerio de Econom\'{\i}a y Competitividad grant PID2020-117354GB-I00 is gratefully acknowledged. The author thanks to On\'esimo Hern\'andez--Lerma and Diego Moreno
for helpful comments. The usual caveat applies.}
\date	{July 30, 2024}
					% Activate to display a given date or no date
\begin{abstract} We analyze the question of whether the outcome of an oligopoly exploiting a nonrenewable resource can be replicated by a related monopoly, within the framework of continuous time and Markov Perfect Nash Equilibrium. We establish necessary and sufficient conditions and find explicit solutions in some cases. Also, very simple models with externalities are shown which Nash equilibrium cannot be replicated in a monopoly.
\end{abstract}
%\begin{document}
\maketitle

\section{Introduction}
Margaret Slade established in \cite{Slade} necessary and sufficient conditions for a Nash equilibrium to be the outcome of a single optimization problem, constructed on the basis of a ``fictitious" objective function matching the first order necessary conditions of the original game. She focused on oligopolistic games, motivated by the question about whether individual firms with selfish objetives may behave as a single agent optimizing a single objective function.  In \cite{Slade}, static as well as dynamic games in discrete time were considered. See \cite{GSHL2013-1}, \cite{GSHL2014} for an analysis of the question within the framework of discrete time stochastic dynamic games, and \cite{GSHL2016} for a survey of results.

Recently, the continuous time case has been incorporated to the literature, as in \cite{Dragoneetal}, \cite{FonsecaHernandez18} or \cite{FonsecaHernandez20}. The authors in \cite{Dragoneetal} propose to construct a ``fictitious" Hamiltonian function fulfilling two premises: the first order partial derivatives of the new Hamiltonian and of the original Hamiltonians of the players must coincide, and the new Hamiltonian can be written as the sum of a function (it may depend on the costate variables of the players) and the product of the dynamics of the original game with the corresponding costate variables. The authors show that, under some circumstances, the Hamiltonian potential
can be the representation of the original differential game. The method in \cite{Dragoneetal} can be applied to open--loop or to feedback strategies. The need for working with the costate variables of the original differential games, as well as maintaining the original dynamics, make difficult to find the expression of the fictitious payoff function of the control problem equivalent to the game. 
 \cite{FonsecaHernandez18} and \cite{FonsecaHernandez20}
 set the question within the framework of open loop Nash equilibrium. A problem with open loop strategies is that Nash equilibrium is not credible in general, that is, it is not subgame perfect.\footnote{See \cite{Basar} or \cite{Fudenberg} for excellent accounts about the impact of the information available to the players on the properties of Nash equilibrium. Here, we focus on the class of Markov strategies, where the players have access to the current value of time and state variable prior to taking decisions.} The oligopolist may have incentives to deviate from the equilibrium at intermediate stages of the game. It is thus desirable to tackle the question addressed  in \cite{Slade} allowing the players to use Markov strategies.  

It is not unreasonable to think that, as the players are more strategically sophisticated, it becomes harder to summarize their strategic behaviour into a single agent optimization problem. In fact, this is the case, as we will show along the paper. Uncorrelated payoffs can also make difficult to answer the question. In general, the more uncorrelated the players' payoffs are, the more difficult is for a single--agent problem to be observationally equivalent to the oligopoly. For instance, we will show that in a duopoly with asymmetric externalities, there are preferences of the players that cannot be subsumed into the preferences of a unique agent\footnote{Note that team problems have players perfectly positively correlated and zero sum games are perfectly negatively correlated. As it was shown on Slade's paper, these two extreme cases can be represented by preferences of a unique player.}

It is worth noting that, although the fictitious monopoly version of the oligopoly game maintains the same structure, in general it has different payoff and different dynamics. The ``fictitious" payoff was already in the foundational Slade's paper\footnote{The construction of the fictitious payoff gave rise to the fundamental concept of potential  game, see \cite{MondererShapley}.}
, but realizing the need of a ``fictitious" dynamics, is new, to our knowledge. This is in contrast with the open loop Nash equilibrium, where the dynamics may be maintained, see \cite{Dragoneetal}, \cite{FonsecaHernandez18} and \cite{FonsecaHernandez20}.

An interesting byproduct of constructing a monopoly being observationally equivalent to a given oligopoly with Markov players is that it induces a way to measure the effect of oligopoly competition on resource preservation. This is somewhat related with the tragedy of the commons phenomenon, see for instance \cite{Clemhout86} and \cite{Clemhouthandbook}. The tragedy of the commons appears when Nash equilibrium implies higher depletion of the resource than under cooperative exploitation. There is a problem with measuring the tragedy of the commons impact: it needs comparison of Nash equilibrium with one of the Pareto optimal solutions of the game. It arises the question of which of the multiple Pareto solutions to choose for comparison. In the symmetric case, the equal--weight Pareto solution is a quite natural choice, but in the asymmetric case it seems to be not so straightforward. In the particular games that we study in this paper---oligopolies of resource extraction--- we define a competition index of the oligopoly as the difference on the total extraction rate under oligopoly and the fictitious monopoly, divided by the total extraction rate under oligopoly. The higher the value of the index, the higher the level of competition for the resource among oligopolist due to uncorrelated objectives. For instance, in a symmetric oligopoly of resource extraction of $N$ players and without externalities, the index is
$$
\text{CI} = \frac{N-1}N \frac1{\varepsilon(u)},
$$
where $u$ is the individual extraction rate, $N$ is the number of players and $\varepsilon(u)$ is the elasticity of the marginal utility. See the definition of the ``Competition Index" below in \eqref{CI}. Since $\varepsilon >0$ when the utility function is increasing and concave and there is no externalities, we see that CI is positive and hence increasing with $N$. If the utility is HARA, then $1/\varepsilon(u)=1/\alpha>0$ and then CI is constant. Since $1/\alpha$ is also the relative index of risk aversion, it happens that CI is greater in symmetric oligopolies where the player's risk aversion is greater, meaning that oligopolists compite more, the more risk averse they are. Externalities put another dimension on the index. In the case of symmetric Cobb--Douglas preferences with externalities and where the individual utility function is given by $L^i(u^1,\ldots,u^N)=(1-\alpha)^{-1}(u^i)^{1-\alpha}\prod_{j\neq i} (u^j)^{1-\beta}$, with $\alpha, \beta >0$, the index becomes---see \eqref{CI}---
$$
\text{CI} = \frac{N-1}N \left(\frac{\frac{\alpha - \beta}{1-\alpha}}{-\alpha + (N-1)(1-\beta)}\right).
$$
With $\beta = 1$ (no externality), the expression above for CI is recovered, whilst $\beta >1$ ($0<\beta<1$) means that the externality is negative (positive). Note that when $\alpha = \beta$, the index equals 0, meaning that the extraction effort is the same in the oligopoly and in the equivalent monopoly. For other values of $\alpha $ and $ \beta$, the competition strength in the oligopoly depends on the number of players $N$ and on $\alpha$ and $\beta$.

The organization of the paper is as follows. In Section 2 we establish the framework where we develop our results, that consists of a differential game with $N$ players, smooth data and no constraints---or with constraints, but with interior Nash equilibrium---.  In Section 3 we set a control problem of the same structure that may be capable of rationalizing the differential game. In Section 4, we focus on oligopolies of extraction of a nonrenewable resource
%as in \cite{Clemhout}, \cite{Clemhout86}, \cite{Clemhouthandbook},
with externalities. Theorem 1 in Section 5 shows that the symmetric oligopoly is always reproducible as a monopoly, with suitable fictitious payoff, discount factor, dynamics and bequest function. Section 6 applies the results above to find explicit expressions in two different oligopoly modes, the difference being that in one of the models the market price is made explicit, whereas in the other model it is hidden into the preferences of the agents, which may show actitudes towards consumptions of the competitors, interpreted as externalities.
%The result is based on a simple algebraic argument appearing in Proposition \ref{Prop:thetaODE} in Appendix \ref{AppA}.
Section \ref{Conclusions} establishes the conclusions. The paper has three appendixes. Appendix \ref{AppA} explains the methodology that we use in our investigation. It consists in using a system of partial differential equations obtained from the Hamilton--Jacobi--Bellman equations, characterizing directly the feedback Nash equilibrium. Appendixes \ref{AppB} and \ref{AppC} study asymmetric oligopolies, based on
Lemma \ref{lemma:Theta} in Appendix \ref{AppA}, which establishes a necessary condition for observational equivalence in asymmetric games. This is used in Appendix \ref{AppB} to solve the problem with a duopoly with multiplicative preferences with externalities, in an infinite horizon. Appendix \ref{AppC}
focusses on additive externalities. It is shown in Theorem \ref{th:NO} of this appendix that there are preferences such that the outcome cannot be replicated by a single agent problem. On the positive side, sufficient conditions are given to make possible the representation in a particular class of duopolies of finite horizon and null discount factor.

It is important to say that along the paper it is taken for granted that in all the games studied a MPNE exists, at least, for some values of the parameters.

\section{Description of the differential game, Nash Equilibrium and Hamilton--Jacobi--Bellman equations}
The paper makes extensive use of notation and results exposed in Appendix \ref{AppA}\footnote{Notation: A subscript indicates partial
differentiation; derivatives of real functions are indistinctly notated with subscripts or with primes; the partial derivative of a scalar function with
respect to a vector and the partial derivative of a vector
function with respect to a scalar are defined as column vectors.
Also, the partial derivative of a vector function with respect to
another vector is defined as a matrix, e.g. $\,h_z=\partial
h/\partial z = (\partial h^i/\p z^j)_{n\times m}$, where $\,h\,$
and $\,z\,$ are $\,n\times 1\,$ and $\,m\times 1\,$ vectors,
respectively. The symbols $\partial_t$ and $\partial_x$ denote total derivatives (or derivatives of a compound function). The superscript $\,^{\top}\,$ will denote the transposition
}
. 

We consider an $\,N$--person differential game on a time interval $[0,T]$ (finite horizon) or $[0,\infty)$ (infinite horizon). 
For $i=1,\ldots, N $, player $i$ chooses a vector of strategies
$u^ i = (u^ i_1,\ldots,u^i_n)^{\top}\in U^i\,$ which affect the
vector of state variables, $y=(y^1,\ldots,y^n)^{\top} \in  \R^n$,
which evolution is given by the system of ordinary differential
equations
\begin{equation}\label{Sisa}
\dot y(s) =F(s, y(s), u^ 1(s),\ldots, u^ N(s)),\quad t\le s\le T,
\end{equation}
with initial condition
\begin{equation}\label{Sisb}
y(t)=x,\quad t\in[0,T],\quad x\in X\subseteq \R^n.
\end{equation}
The control region of $i$th player is denoted $U^i\subseteq \R^n$.
The data $(t,x)$ is the initial node or root of a subgame, that we will identify with the subgame itself. The bequest function is the final payoff of the player when the game ends. When the time horizon is infinite, there is no bequest function in this case.

\begin{definition}\label{Def}{\rm (Admissible strategies)}.
A strategic profile $\,u=(u^1,\ldots, u^N)\,$ is called admissible
if $u(s)\in U^1\times\cdots \times U^N$ for every $s\in[0,T]$ and
\begin{enumerate}
\item[{\rm (i)}] for each $i=1,\ldots, N$, there exists a function
$\,\phi^i : [0,T]\times X\subseteq\mathbb{R}^n\longrightarrow U^i\,$ of
class $\mathcal{C}^{1}$ such that $\,u^i(s) = \phi^i(s,y (s) )$
for every $s\in [0,T]$;
 \item[{\rm (ii)}] for every $(t,x)$ the system \eqref{Sisa} with
initial condition $y(t)=x$ admits a unique solution.
\end{enumerate}
\end{definition}
Thus, we consider feedback strategies. Let $\,{\mathcal U}^i$ be the set of this kind of admissible strategies of
player $i$ and let $\, {\mathcal U}={\mathcal U}^1\times \cdots \times {\mathcal
U}^N$ be the set of admissible strategy profiles. If $\,\phi^i\,$ is
time independent, the corresponding control is called a stationary
Markov control.
To simplify the notation, we will identify $u^i$ with the feedback rule $\phi^i$.

The instantaneous utility function of player $i$ is denoted by
$L^i$ and his or her bequest function by $B^i$. Given $\,(t,x)\in
[0,T]\times X\,$ and an admissible strategic profile $\,u$, the
payoff functional of each player is given by
\begin{equation}\label{Fun}
J^i(t,x;u)=\int_t^T e^{-r_i(s-t)}L^i(s,y(s),u(s))\,ds +
e^{-r_i(T-t)}B^i(T,y(T)),
\end{equation}
with $r_i\ge 0$ the discount rate. $J^i$ is the utility obtained by player $i$ when the games starts
at $(t,x)$ and the profile of strategies is $u$.

As said above, in the infinite horizon case the time interval is $[0,\infty)$ and the
bequest function is null. In this case, if the problem
is autonomous and the strategies are Markov stationary, the payoff
functionals are independent of time, and the initial condition is
simply $x$, with $t=0$.

The functions
\begin{align*}
F&:[0,T]\times X\times U^1\times \cdots \times U^N\longrightarrow X,\\
L^i&:[0,T]\times X\times U^1\times \cdots \times U^N \longrightarrow \R,\\
B^i&:[0,T]\times X\longrightarrow \R,
\end{align*}
are all assumed to be of class $\mathcal{C}^{2}$.

In a non--cooperative setting the aim of the players is to
maximize their individual payoff $J^i$. Since this aspiration
depends also on the strategies selected by the other players, it
is generally impossible to attain. An adequate concept of solution
is Nash equilibrium, which prevents unilateral deviations of the
players from its recommendation of play.
The Markov Perfect Nash Equilibrium (MPNE) considers optimality at every subgame $(t,x)$.

\begin{definition} {\rm(MPNE)}.
An $\,N$-tuple of strategies $\,u^{*}\in {\mathcal U}\,$ is
called a Markov perfect Nash equilibrium if for every $\,(t,x)\in
[0,T]\times X$, for every $u ^ i\in {\mathcal U}^i$
\[
J^ i\big(t,x;\phii \big)) \le J^ i(t,x;u^*),
\]
for all $\,i=1,\ldots,N$.
\end{definition}
In the definition above, $\phii$ denotes $
(u^{*1},\ldots,u^{*i-1},u^i,u^{*i+1},\ldots,u^{*N})$.
Note that at an MPNE, no player has incentives to deviate
unilaterally from $u^*$, whatever the initial condition
$\,(t,x)\,$ is.
Let $\,V^i\,$ be the value function of the $\,i$th player, that is
\[
V^i(t,x) = \max_{u ^i\in {\mathcal U}^i} \left\{J^i(t,x;\phii)\;:\;
\dot y = F(s,y,\phii ),\ y(t) = x, \forall s\in (t,T)\right\}.
\]
Under our smoothness conditions, the value functions satisfy the HJB system of PDEs
\begin{equation}\label{HJBgame}
-r_iV^i + V^i_t + \max_{u ^i\in {\mathcal U}^i} \left\{e^{-r_it}L^i(t,x,\phii) + F^{\top}(t,x,\phii)V_x^i \right\}= 0,\quad i=1,\ldots,N.
\end{equation}
From this famous PDEs and upon taking derivatives, we will show in Appendix \ref{AppA} another system of PDEs characterizing Nash equilibrium. Although this has been obtained previously in other papers by the author and collaborators, we will repeat the exercise in the appendix for convenience of the reader.

\section{Question addressed in the paper and equivalent optimal control problem}

The problem addressed in this paper is the following:\it

Given the $N$--person noncooperative game described in the previous section
\begin{equation}\label{DG}
\mathrm{DG} = ((L^i)_{i=1}^{N}, (r_i)_{i=1}^{N}, F, (B^i)_{i=1}^{N},(\mathcal U^i)_{i=1}^{N})
\end{equation}
with MPNE $u^*\in \mathcal{U}$, determine an optimal control problem
\begin{equation}\label{OC}
\mathrm{OC} = (\ell, \rho, f, b),
\end{equation}
which admits $u^*$ as optimal solution.

\rm

In OC, $\ell$ is the payoff integrand, $\rho\ge 0$ is the discount factor, $f$ is the dynamics, and $b$ is the bequest function.
We exclude from the description of OC the feasible set $\mathcal{U}=\mathcal{U}^1\times \cdots\times\mathcal{U}^N$, as it is given in the description of the original DG, and it does not change.

\begin{definition}\label{equi}
We will say that the single agent decision problem OC is equivalent to the differential game DG, or that OC rationalizes the MPNE $u^*$, if $u^*$ is solution of both DG and OC.
\end{definition}

The (equivalent) optimal control problem's full description is
\[
\max_{u\in \mathcal{U}}  J(t,x;u)=\int_t^T e^{-{\rho}(s-t)}\ell(s,z(s),u(s))\,ds +
e^{-\rho(T-t)}b(T,z(T)),
\]
subject to
\[
\dot z(s) =f(s, z(s), u^ 1(s),\ldots, u^ N(s)),\quad t\le s\le T,
\]
with initial condition
$$
z(t)=x,\quad t\in[0,T],\quad x\in X\subseteq \R^n.
$$
Of course, the maximization is performed at once with respect to $u^1,\ldots,u^N$, and $u^*$ is an optimal feedback control iff $J(t,x;u^*)\ge J(t,x;u)$, for all admissible $u$ and all admissible initial condition $(t,x)$.

The functions
\begin{align*}
f&:[0,T]\times X\times U^1\times \cdots \times U^N\longrightarrow X,\\
\ell&:[0,T]\times X\times U^1\times \cdots \times U^N \longrightarrow \R,\\
b&:[0,T]\times X\longrightarrow \R,
\end{align*}
are all assumed to be of class $\mathcal{C}^{2}$.

\section{The oligopoly game}

We analyze the equivalence question which motivates this paper, which has been established in Definition \ref{equi}, in the framework of a dynamic oligopoly game of resource extraction in continuous time, where the competition for a resource may be affected by consumption externalities. This means that the oligopolists have preferences defined not only on their own consumption, but also on the consumption of the other players. More specifically, 
oligopolist $i\in\{1,\ldots,N\}$ has preferences given by a utility function $L^i(u^i,u_{-i})$, where $L^i$ is symmetric with respect to the consumption of the rest of the players, $u_{-i}=(u^1,\ldots,u^{i-1},u^{i+1},\ldots,,u^N)$. Symmetry means that the consumption of the other players deserves the same preference consideration to the agent\footnote{The case where it matters who of the the other players is consuming the resource could be analized analogously; this requires of course $N> 2$.}: if $\pi u_{-i}$ is a permutation of the profile $u_{-i}$ (that is, a one-by-one exchange of the indexes of the rest of players), then $L^i(u^i,u_{-i})=L^i(u^i,\pi u_{-i})$. Here, $u^i$ is the extraction or consumption rate of a non renewable resource which stock at time $t$ is $y(t)$. Thus,  $\dot y = -\sum_{j=1}^N u^j$ is the evolution of the resource. We wish to study whether it is possible that the MPNE of this game with possible externalities (assuming existence) is identical to a monopolistic problem with suitable preferences, discount factor,  dynamics and bequest function. We assume that the utility of player $i$, $L^i$, is increasing and strictly concave in own consumption, that is
$L^i_{u^i} >0$ and $L^i_{u^iu^i} <0$.
The preferences are given by the functional
\begin{equation}\label{preferences}
J^i(t,x,u^i,u_{-i}) = \int_{t}^T e^{-r_i(s-t)}\,L^i(u^i,u_{-i})\, ds + e^{-r_i(T-t)}B^i(y(T))
\end{equation}
where $r_i\ge 0$ is the discount factor and, as said above
\begin{equation}\label{dynamics}
\dot y = - u^i-\sum_{j\neq i} u^j,\quad y(t)=x\ge 0.
\end{equation}
The case $T=\infty$ is also allowed, which presupposes that there is not bequest payoff at $\infty$, $B^i\equiv 0$.

Let
\begin{equation}\label{Ei}
E_i(u^1,\ldots,u^N) = -\frac{L^i_{u^i}}{\sum_{j=1}^N L^i_{u^iu^j}}(u^1,\ldots, u^N) > 0;
\end{equation}
\begin{equation}\label{E-i}
E_{-i}(u^1,\ldots,u^N) = -\frac{\sum_{j\neq i} L^i_{u^j}}{(N-1)\sum_{j=1}^N L^i_{u^iu^j}}(u^1,\ldots, u^N).
\end{equation}
We define $E_{-i}=0$ if $N=1$. We assume that $E_i>0$, but the sign of $E_{-i}$ is not fixed in advance.

In the case in which $L^i_{u^iu^j}=0$ for $j\neq i$, $E_i$ is the inverse of the absolute risk aversion index of Arrow-Pratt with respect to player's own consumption. In the general case, it can be considered as an average index risk that takes into account consumption of the other players. 
A similar interpretation can be given to $E_{-i}$, as a measure of risk about the consumption of the other players. When $E_{-i}>0$, an increase in consumption of player $j\neq i$ (and thus, of any other player, given the assumption about symmetry made on $L^i$ about $u_{-i}$) raises utility of player $i$, implying that increasing the consumption level of the other players is a positive externality. This feature of the preferences is known in the literature as ``keep up with the Joneses", see \cite{Abel} and \cite{Gali}.
If $E_{-i}<0$, then consumption of the rest of agents diminishes own utility, since the marginal utility is negative, so that $u^j$ can be considered as a substitute good to $u^i$, for $j\neq i$. In this case, consumption of other agents is a negative externality.

\section{Construction of the equivalent control problem to the oligopoly game}

Our approach to the problem rests on a system of PDEs arising as necessary conditions for solutions both of the original game and of the equivalent optimal control problem. This approach was developed in \cite{GJP05}, \cite{JPJG98}, \cite{JP04} for the deterministic case and \cite{RJP07}, \cite{RJP15} for stochastic problems. In the problem we study here, the only element that remains invariant both for the game and for the equivalent control problem is the solution, $u^*$; the rest of elements, like the payoffs, the discount factor, the dynamics and the bequest functions may change, as well as the value functions. Thus, working only with the HJB equations could be not sufficient to solve the problem, being more useful to work directly with the PDEs characterizing the MPNE, $u^*$, which is the invariant object. The fact that two different systems of PDEs (that for the differential games and that for the control problem) admit the same solution is a strong condition that can be applied to get insights into the problem. The results that follow make extensive use of Appendix \ref{AppA} and hence we will refer to that section of the paper most often.

Our starting point is obtain the PDE system \eqref{PDE_MPNE} in Appendix \ref{AppA} for the oligopoly game, by taking total derivatives in \eqref{PDE_conservative}, noting that
$\Gamma^i=L^i_{u^i}$ and $\Ha^i=L^i-(\sum_{j=1}^N u^j)L^i_{u^i}$ for this game, see Remark \ref{PDEconservative}.
The system becomes
\begin{align}\label{PDE}
-r_iL^i_{u^i} + \sum_{j=1}^N L^i_{u^iu^j}u^j_t -\left(\sum_{j=1}^N u^j\right) L^i_{u^iu^i} u^i_x + \sum_{j= 1, j\neq i}^N \left(
L^i_{u^j}-L^i_{u^i} -\left(\sum_{j=1}^N u^j\right) L^i_{u^iu^j} \right)u^j_ x = 0,
\end{align}
for $i=1,\ldots, N$. We have abbreviated notation, with $u^i_t$ meaning $u^i_t(t,x)$ and $u^i_x$ meaning $u^i_x(t,x)$. Of course, $u^1(t,x),\ldots,u^N(t,x)$ are the unknowns of the PDE system and the MPNE of the game $({u^*}^1(t,x),\ldots,{u^*}^N(t,x))$ is a solution.

First we study the symmetric game, both in the finite and the infinite horizon case. We will work out a duopoly version of the asymmetric game in the appendixes \ref{AppB} and \ref{AppC}.

\subsection{Symmetric oligopoly game}\label{Symmetric}
Assume that the game is symmetric, that is
$$
L^i(u^i,u_{-i}) = L^j(u^j,u_{-j}),\quad r_i=r,\quad B^i\equiv B\ \mbox{(when $T<\infty$)},
$$
for all $i, j=1,\ldots,N$ and that we look for symmetric Nash equilibrium $(u^1,\ldots, u^N)=(u,\ldots,u)$. Then \eqref{PDE}  
reduces to a single equation. It suffices to focus on the first player, say. In what follows, let
$$
n_j(u) = L^1_{u^j}(u,\ldots,u),\quad n_{1j}(u) = L^1_{u^1u^j}(u,\ldots,u).
$$
Due to the postulated symmetry, the first PDE in \eqref{PDE} becomes
$$
\left(\sum_{j=1}^Nn_{1j}(u)\right)  u_t + \left(-Nu \left(\sum_{j=1}^Nn_{1j}(u)\right) + \left(\sum_{j\neq i}^Nn_{j}(u)\right) - (N-1) n_1(u)\right)u_x = r n_1(u).
$$
or
\begin{equation}\label{game}
u_t +(-Nu + (N-1) e_1(u) -  (N-1)e_{-1}(u) ) \, u_x = -r e_1(u),
\end{equation}
where $e_1(u) = E_1(u,\ldots,u)$ and  $e_{-1}(u)=E_{-1}(u,\ldots,u)$. See equations \eqref{Ei} and \eqref{E-i}. Remember that $e_{-1}=0$ if $N=1$.

If $u$ is also the optimal solution of a control problem with payoff $\ell(u)$, dynamics $f(u)$, discount $\rho$ and bequest function $b(x)$ ($b$ is needed only when the game horizon is finite), then $u$ must satisfy the FOCs \eqref{OChamiltonian},  \eqref{OCcoestate} and \eqref{PDE_OC} in Appendix \ref{AppA}, which become in this context
\begin{align}\label{control}
u_t + f(u) u_x &= \rho \frac{\gamma(u)}{\gamma'(u)}\\
u(T,x)&=\varphi(x)\quad  \mbox{(if $T<\infty$)},
\end{align}
where $\gamma(u) =-\frac{\ell'(u)}{f'(u)}$.
So, comparing \eqref{game} and \eqref{control}, we can identify\footnote{Of course, two PDEs may share the same set of solutions but still not be proportional with a non null factor. Frobenius Theorem of integrability of PDEs plays a role here; it was used in \cite{JPGJ00} and in \cite{GJP05} for the determination of MPNE which are at the same time Pareto optimal in deterministic games. We implement the simplest case that both equations are the same. It is worth noting that the dynamics of the equivalent monopoly is not $-Nu$, which is the dynamics of the symmetric oligopoly game.}
\begin{equation*}\label{dynamicscontrol}
f(u) = -Nu +  (N-1) (e_1(u)  -  e_{-1}(u));
\end{equation*}
and
\begin{equation}\label{gamma0}
\rho \frac{\gamma(u)}{\gamma'(u)} = -r e_1(u).
\end{equation}
Integrating this differential equation for $\gamma$, we obtain
\begin{equation}\label{gamma}
\gamma(u)=Ce^{-\frac{\rho}r\int \frac {du}{e_1(u)}},
\end{equation}
where $C$ is a constant. Plugging this into the identity $\ell'(u)=-f'(u)\gamma(u)$ and integrating by parts
\begin{equation*}\label{preferencescontrol}
\ell(u) = Cf(u)e^{-\frac{\rho}{r}\int \frac{du}{e_1(u)}} + C \frac{\rho}{r}\int \frac{ f(u)}{e_1(u)} e^{-\frac{\rho}{r}\int^u\frac{dz}{e_1(z)}}\, du.
\end{equation*}
Thus, we have found a triplet OC$=(\ell,\rho,f)$ that is a candidate to rationalize the MPNE $u^*$ in the infinite horizon case. Any $\rho>0$ works, including $\rho=0$ in the finite horizon case.\footnote{
A complete specification will require to check monotonicity and concavity of preferences. For this, it is useful to note that 
$\ell'(u) = -\gamma(u) f'(u)$ and hence $\ell''(u)=-\gamma'(u)f'(u) - \gamma(u) f''(u)$ and study the signs.
Since this does not pose a major conceptual problem, we skip the details.}

%\end{remark}
\subsubsection{
The case with finite horizon, $T<\infty$.} When the oligopoly is of finite duration, to the triplet OC$=(\ell,\rho,f)$ found above it is needed to incorporate a suitable bequest function, $b$. This can be done by observing that there is a final condition for the MPNE at the terminal time $T$, see \eqref{curva} in Appendix \ref{AppA}
\begin{equation}\label{FCgame}
\left.L^i_{u^i}(u^1(t,x),\ldots,u^N(t,x)) - B_x^i(x)\right|_{t=T} = 0,\quad i=1,\ldots, N.
\end{equation}
Assuming that this system defines univocally the value of the MPNE at $(T,x)$, we write
\begin{equation}\label{FCu}
u^i(T,x)=\varphi^i(x),\quad i=1,\ldots, N,
\end{equation}
for a suitably smooth function $\varphi^i$. In the symmetric game we are analyzing in this section, $\varphi^i\equiv \varphi$.

As is \eqref{FCgame}, the final condition for the equivalent control problem becomes
$$
\left.\ell'(u(t,x)) - b_x(x) f'(u(t,x))\right|_{t=T} = 0,
$$
where $b$ is the bequest function for the control problem we are looking for. Since $\ell'(u)=-\gamma(u)f'(u)$, and $f'(u)\neq 0$, we have
$
b_x(x) = -\gamma (u(T,x))
$.
Change variable $u=\varphi(x)$ into the integral $\int \frac {du}{e_1(u)}$ to obtain
$$
\int \frac {du }{e_1(u)}= \int  \frac {\varphi_x(x)}{e_1(\varphi(x) )}\, dx 
$$
and hence from, \eqref{gamma}
$$
\gamma (u(T,x)) = C e^{(-\frac{\rho}r)\int  \frac {\varphi_x(x)}{e_1(\varphi(x) )}\, dx }.
$$
Thus, the bequest function $b$ for the equivalent control problem has derivative
\begin{equation*}\label{bx}
b_x(x) = C e^{(-\frac{\rho}r)\int  \frac {\varphi_x(x)}{e_1(\varphi(x) )}\, dx }.
\end{equation*}
A further integration recovers $b$.

We collect the above considerations in the following theorem.
\begin{theorem}\label{ThMain}
The symmetric oligopolistic game can be rationalized as a fictitious monopolistic problem OC=$(\ell,\rho,f,b)$ where
\begin{equation}\label{preferencescontrol}
\ell(u) = Cf(u)e^{-\frac{\rho}{r}\int \frac{du}{e_1(u)}} + C \frac{\rho}{r}\int \frac{ f(u)}{e_1(u)} e^{-\frac{\rho}{r}\int^u\frac{dz}{e_1(z)}}\, du;
\end{equation}
\begin{equation}\label{dynamicscontrol}
f(u) = -Nu +  (N-1) (e_1(u)  -  e_{-1}(u));
\end{equation}
$$
\rho\ge 0\quad(\rho>0 \mbox{ if $T=\infty$});
$$
\begin{equation}\label{bx}
b(x) = C \int e^{(-\frac{\rho}r)\int  \frac {\varphi_x(z)}{e_1(\varphi(z) )}\, dz}\,dx \mbox{ if $T<\infty$ and $b=0$ if $T=\infty$.}
\end{equation}
$C$ denotes a constant.
\end{theorem}

%\begin{remark}[Competition Index] \rm
In the trivial case $N=1$, obviously OC=DG. When $N\ge 2$, the total extraction rate in the symmetric DG is $Nu$, but in the fictitious monopoly, it is, according to \eqref{dynamicscontrol}
$$
U_{\mathrm{mon}}\equiv -f(u)=Nu - \left(N-1) (e_1(u) -  e_{-1}(u)\right).
$$
Whether $U_{\mathrm{mon}}$ is greater or smaller than $Nu$ depends on the sign of $e_{1}(u) - e_{-1}(u)$. If $e_1>0$ and the externality is negative, $e_{-1}(u)<0$, then $e_{1}(u) - e_{-1}(u)>0$. In this case, $U_{\mathrm{mon}}<Nu$, meaning that the total extraction rate in the fictitious monopoly is smaller than the total extraction rate in the oligopoly. Then we can interpret that negative externalities make the competition stronger than with positive externalities. This interpretation is based on the speed rate to which the resource is depleted under the oligopoly or under fictitious monopoly play.
\begin{definition} The Competition Index of the oligopoly is defined by\footnote{Although this definition is placed for symmetric oligopolies, it is pretty clear that it can be translated to asymmetric ones in the obvious way.}
\begin{equation}\label{CI}
\mbox{$\mathrm{CI}$}= \frac{Nu - U_{\mathrm{mon}}}{Nu} = \left(\frac{N-1}N\right) \left(\frac { e_{1}(u) - e_{-1}(u) }{u}\right).
\end{equation}
\end{definition}
The higher the CI is, the more intense the competition in the oligopoly is.
We will illustrate this index with examples in Section \ref{externalities}.

\section{Applications}
In this section we apply Theorem \ref{ThMain} above to two different formulations of the oligopoly. In the first one, the market price of the resource is given by an inverse demand function, whereas in the second formulation, the oligopolist gets utility from the use of the resource, without being apparent a price. This utility may be affected by the other oligopolists' behavior, which we interpret as an externality. The preferences are of multiplicative type in this case.

\subsection{Oligopoly pricing}

Suppose that the market where the resource is traded is defined by an inverse demand function $p(Q)$, where $Q=\sum_{j=1}^Nu^j$ is the total extraction effort in the industry. Suppose that the cost of extraction of the resource is $c(u^i)$. Both $p$ and $c$ are smooth functions. We only consider inthe finite horizon case to keep the paper within reasonable bounds. The profit of oligopolist $i$ is
$$
L^i(u^i,u_{-i})=u^ip(Q)-c(u^i),
$$
which is the integrand payoff function of the game of player $i$. The following computations are straightforward
\begin{align*}
L^i_{u^i} &= p(Q)+u^ip'(Q)-c'(u^i),\\
L^i_{u^j} &=u^i p'(Q),\quad j\neq i\\
L^i_{u^iu^i} &= 2p'(Q) + u^i p''(Q)-c''(u^i),\\
L^i_{u^iu^j} &= p'(Q)+u^ip''(Q),\quad j\neq i.
\end{align*}
Then
\begin{align*}
E_i(u^1,\ldots,u^N)&=-\frac{p(Q)+u^ip'(Q)-c'(u^i)}{2p'(Q) + u^i p''(Q)-c''(u^i) + (N-1)(p'(Q) + u^ip''(Q))},\\
E_{-i}(u^1,\ldots,u^N)&=-\frac{u^i p'(Q)}{2p'(Q) + u^i p''(Q)-c''(u^i) + (N-1)(p'(Q) + u^ip''(Q))}.
\end{align*}
The oligopoly is symmetric, and thus we consider symmetric MPNE $u^1=\ldots = u^N= u$. Under symmetry, we have $Q=Nu$ and corresponding expressions $e_i(u)=E_i(u,\ldots,u)$, $e_{-1}(u)=E_{-1}(u,\ldots,u)$.
Theorem 1 above gives a positive answer to the observational equivalence question. To find closed form solutions suppose that $p(Q)=AQ^{-q}$, where $A>0$ and $q>0$, and eliminate costs. Then $p(Q)=AN^{-q}u^{-q}$ and
%\begin{align*}
%e_1(u)&=-\frac{k_1}{k_2}u-\frac{c}{k_2}u^{1+q},\\
%e_{-1}(u)&= \frac{AqN^{-1-q}}{k_2} u,
%\end{align*}
\begin{align*}
e_1(u)&=\frac1q\, u,\\
e_{-1}(u)&= k_1 \,u,
\end{align*}
%$k_1=AN^{-1-q}(N-q)$ and
where $k_1=-\frac{N}{-q-1+(N-1)(N-q-1)}$.

Applying Theorem \ref{ThMain} we find
$$
f(u) = -Nu + (N-1)(q^{-1}- k_1)\, u
$$
Let us denote $k_2=-N+(N-1)(q^{-1} - k_1)$. Noting that
$$
\int \frac{du}{e_1(u)} = q \int \frac{du}u = q\ln{u},
$$
after simplifications and collection of common terms, we arrive to the expression
$$
\ell (u) = Ck_2\left(1+\frac{q\frac{\rho}r}{1-q\frac{\rho}r}\right)u^{1-q\frac{\rho}r}.
$$
Thus, the fictitious monopoly consists of a market given by the inverse demand function
$$
\mbox{Constant }\times \ u^{-q\frac{\rho}r}.
$$

%Let us compute the Competition Index. It is 
%$$
%\mbox{CI} = \frac{Nu - U_{\mbox{\scriptsize mon}}}{Nu} = \frac{Nu - f(u)}{Nu} = \frac{(N-1)(N-2q)}{q(-2N+q+1+(N-1)(-N+q+1))}.
%$$

\subsection{Multiplicative preferences with externalities}\label{externalities} In this section we consider oligopolies where a price system is not explicitly given in the specification of the preferences of the players. Instead, we interpret rate of extraction as consumption, and preferences are given by own consumption and consumption of the rest of the players. Hence we let
$$
L^i(u^i,u_{-i}) = m_i(u^i) \prod_{l\neq i} k(u^l).
$$
Thus, the externality affecting an individual player preferences' is multiplicative, $\prod_{l\neq i} k(u^l)$.
We obtain the derivatives
$$L^i_{u^i} = m'_i(u^i) \prod_{l\neq i} k(u^l),\qquad L^i_{u^iu^i} = m''_i(u^i) \prod_{l\neq i} k(u^l)$$
and
$$
L^i_{u^iu^j} = m'_i(u^i) \left(\prod_{l\neq i,j} k(u^l)\right) k'(u^j),
$$
for $j\neq i$.
Thus
$$
E_i=-\frac{L^i_{u^i}}{\sum_{j=1}^N L^i_{u^iu^j}} = \frac{-m'_i(u^i)\prod_{l\neq i} k(u^l)}{ m''_i(u^i) \prod_{l\neq i}  k(u^l)+ m_i'(u^i)\sum_{j=1,j\neq i}^n\left(\prod_{l\neq i,j} k(u^l)\right) k'(u^j)},
$$
$$
E_{-i} = -\frac{\sum_{j\neq i} L^i_{u^j}}{(N-1)\sum_{j=1}^N L^i_{u^iu^j}} = \frac{-m_i(u^i)\sum_{j\neq i} \left(\prod_{l\neq i,j} k(u^l)\right) k'(u^j)}{(N-1)\left( m''_i(u^i) \prod_{l\neq i} k(u^l) + m_i'(u^i)\sum_{j= 1, j\neq i}^n\left(\prod_{l\neq i,j} k(u^l)\right) k'(u^j)\right)}.
$$
In the symmetric case that we are studying in this section, $m_1=\cdots = m_N= m$. In a symmetric equilibrium $u^1=\cdots=u^N=u$, hence we have
$$
e_1(u) = \frac{-m'(u)k(u)^{N-1}}{m''(u)k(u)^{N-1}+(N-1)m'(u)k(u)^{N-2}k'(u)} =  \frac{-m'(u)k(u)}{m''(u)k(u)+(N-1)m'(u)k'(u)}
$$
and
$$
e_{-1}(u) = \frac{-(N-1)m(u)k(u)^{N-2}k'(u)}{(N-1)\left(m''(u)k(u)^{N-1}+(N-1)m'(u)k(u)^{N-2}k'(u)\right)}=\frac{-m(u)k'(u)}{m''(u)k(u)+(N-1)m'(u)k'(u)},
$$
where $e_1$ and $e_{-1}$ where defined above.%in Section \ref{Symmetric}.

%Preferences $L^1(u,v)=-e^{au}e^{cv}$, $L^2(u,v)=-e^{bu}e^{cv}$.
%$$
%A=\left(\begin{array}{cc} -\frac{c\,\left(b-c\right)}{b\,\left(a\,b-c^2\right)} & \frac{b\,\left(a-c\right)}{a\,\left(a\,b-c^2\right)}\\ \frac{a\,\left(b-c\right)}{b\,\left(a\,b-c^2\right)} & -\frac{c\,\left(a-c\right)}{a\,\left(a\,b-c^2\right)} \end{array}\right)
%$$
%$$
%\lambda_{1,2} = \frac 12 \frac 1{ab-c^2} \left(\frac{c(b-c)}b + \frac{c(a-c)}a\pm \sqrt{\left(\frac{c(b-c)}{b} - \frac{c(a-c)}{a} \right)^2+4(b-c)(a-c)}
%\right)
%$$
%Pruebas con Matlab sugieren que siempre son reales
%

\subsubsection{Cobb--Douglas preferences} To get form closed solutions, let
\begin{equation}\label{CDpref}
m(u)=\frac 1{1-\alpha} (u)^{1-\alpha},\quad k(u)=u^{1-\beta},\quad \mbox{with } \alpha>0,\beta >0,-\alpha + (N-1)(1-\beta)<0.
\end{equation}
Following the computations above, we obtain
$$
e_1(u) = \frac{-1}{-\alpha + (N-1)(1-\beta)}u\equiv \eta_ 1\, u.
$$
$$
e_{-1} (u)=\frac{-\frac{1-\beta}{1-\alpha}  }{-\alpha +(N-1)(1-\beta)}u\equiv \eta_ 2\, u.
$$
Then from \eqref{dynamicscontrol} %and \eqref{gamma}
\begin{align*}
f(u) &= -N u + (N-1)( \eta_1 u  - \eta_2 u) 
%\gamma(u)&=Ce^{-(r/\rho)\int du/e_1(u)} = C u^{- r/(\eta_1\rho)}
\end{align*}
and $\ell(u)$ can be recovered easily from \eqref{preferencescontrol},
%$$
%\ell(u) = C\left(\theta + \frac{\frac{r}{\rho}\frac{\theta}{\theta_1}}{1-\frac{r}{\theta_1\rho}} \right)u ^{1-\frac{r}{\theta_1\rho}}.
%$$
%(Comprobar los signos etc.) CASO $\beta=1$.
%$
%\ell(u)=f(u)\gamma(u) + \frac{\rho}r\int \frac{f(u)}{e_1(u)}\gamma(u)\, du
%$. Hence
$$
\ell(u) = Cu^{-m}(-Nu+(N-1)(\eta_1 u - \eta_2u))+C\frac{\rho}r\int \frac{-N+(N-1)(\eta_1-\eta_2)}{\eta_1}u^{-m}\, du,
$$
where $m=\rho/(\eta_1r)$. Note that $\ell(u)$ is a HARA utility with elasticity of the marginal utility $m$.

Regarding the Competition Index \eqref{CI}, it is
$$
\text{CI} = \frac{N-1}N \left(\frac{-1+\frac{1-\beta}{1-\alpha}}{-\alpha + (N-1)(1-\beta)}\right).
$$
In the duopoly case, $N=2$ (following our assumptions, this requires $\alpha + \beta >1$),  CI is $$\frac 12 \frac{\beta-\alpha}{\alpha + \beta -1}.$$
It is positive when $\beta>\alpha$ and negative otherwise. Moreover, if $\alpha >\frac 12$, CI increases with $\beta$.

\section{Conclusion}\label{Conclusions}
Nash equilibrium is not in general Pareto optimal, and thus it is not the solution of an optimal control problem which objective function is a convex combination of the player's payoffs. It is natural then to ask whether Nash equilibrium could be the solution of a different, but related control problem, that is, whether uncoordinated play could be the outcome of cooperative play of a suitable single agent decision problem. From a philosophical point of view, if this is possible, we can interpret that the noncooperative behaviour still can be understood as cooperative in a suitable ``parallel" economic world. It is intriguing that there are very simplistic oligopoly interactions for which there is no such an ideal cooperative parallel word (or fictitious monopoly). This happens when the player's preferences are highly asymmetric and uncorrelated. This work provides a framework to analyze these kind of questions, which were originated in \cite{Slade}, but within continuous time dynamic games when the players use Markov strategies.

\clearpage

\appendix

\section{Technical Results}\label{AppA}

\subsection{A system of partial differential equations for Nash equilibrium}

\strut

To study the question addressed in the paper, that is, whether a MPNE can be attained under the the rule of a unique agent for a suitable specification of the control problem, we 
will resort to the characterization of the MPNE provided in Rinc\'on-Zapatero et al. (1998) (it was afterwards extended  to the non--smooth case in Rinc\'on-Zapatero
(2004) and to the stochastic case in Josa-Fombellida and Rinc\'on-Zapatero (2007, 2015). The characterization consists in a system of PDEs for the MPNE without participation of the value functions. The quasilinear structure of the system will allow us to obtain readily important directions for analyzing the problem. In Rinc\'on-Zapatero et al. (1998) the system for the MPNE was derived from the maximum principle. For convenience of the reader we derive it here too, but from the HJB equations.

Since we are supposing that the Nash equilibrium is interior, we have
$$
H^i_{u^i} (t,x,u^*,V^i_x) = L^i_{u^i}(t,x,u^*)+F^{\top}_u(t,x,u^*)V^i_x = 0.
$$
Let us denote, for any admissible profile $u$
\begin{equation}\label{Gamma}
\Gamma^i(t, x,u)=\left(F^{-\top}_uL^i_{u^i}\right)(t,x,u),
\end{equation}
for $i=1,\ldots, N$. Each vector $\Gamma^i=(\Gamma^i_1,\ldots,\Gamma^i_n)^{\top}$ has $n$ components, where $n$ is the dimension of the state variable $x$. Note that by definition, if the HJB system is fulfilled, then $\Gamma^i_k(t,x,u^*) =V^i_{x_k}(t,x)$, for $k=1,\ldots,n$.
Now, we derive in the HJB equations \eqref{HJBgame} with respect to $x_k$.

Applying the Envelope Theorem, the derivative of 
$
\max_{u ^i\in {\mathcal U}^i} H^i (t,x,(u^i|u^{*}_{-i}),\Gamma^i_k(x,u^*))
$
is
$$
H^i_{x_k} + \sum_{j\neq i} u^{j\top}_{x_k}  H^i_{u^j}  + \sum_{l=1}^n F^l \partial_{x_k} \Gamma^i_l,
$$
where we eliminate the arguments of the several functions involved, and where write $u$ instead of $u^*$ to facilitate the reading. Since $V_{x_kx_l} = V_{x_lx_k} $, we have $\partial_{x_k} \Gamma^i_l = \partial_{x_l} \Gamma^i_k$ and thus the matrix
$$
\left(\sum_{l=1}^n F^l \partial_{x_k} \Gamma^i_l \right)_{1\le l,k\le n}
$$
can be written as the matrix product $\left(\partial_x \Gamma^i \right)F$, for $i=1,\ldots,N$, and where $\partial_x \Gamma^i = \Gamma^i _x +\sum_{j=1}^N \Gamma^i_{u^j}u^j_x$. Again due to the regularity of $V^i$, $V_{tx_k} = V_{x_kt} $, thus $$V_{tx_k} = \partial_ t \Gamma^i_k = \sum_{j=1}^N \Gamma^i_{k, u^j}u^j_{t}$$ or, in vector form, $\partial_ t \Gamma^i = \Gamma_t^i + \sum_{j=1}^N \Gamma^i_{u^j}u^j_{t}$. Hence, the derivative of the HJB equations \eqref{HJBgame} with respect to $x$ is, after rearranging terms
\begin{equation}\label{PDE_MPNE}
-r_i\Gamma^i  + \sum_{j=1}^N \Gamma^i_{u^j}(u^j_{t} + u^j_x F )+ \sum_{j\neq i} u^{j\top}_{x}  H^i_{u^j}  + H^i_x + \Gamma^i _x F =0,
\end{equation}
for $i=1,\ldots,N$, where the Hamiltonian's derivatives are evaluated at
$$
(t, x, u, \Gamma^i(x,u))
$$
and where $u=u^*$ is the MPNE. This is a system of $N\times n$ PDEs for $u^1,\ldots,u^N$, where the value functions do not appear. Notice that each profile $u^i$ has $n$ strategies, and the number of players is $N$.

\begin{remark}\label{PDEconservative}
Since $H^i_{u^i} = 0$, adding $u^{i\top}_{x}  H^i_{u^i} $ to \eqref{PDE_MPNE} does not change the system, but allows us to write it in conservative form as follows
\begin{equation}\label{PDE_conservative}
\partial_t\left(e^{-r_it}\Gamma^i(t,x,u)\right)) + \partial_x \left(e^{-r_it} \Ha^i(t,x,u)\right) = 0,
\end{equation}
for $i=1,\ldots,N$,
where $\Gamma^i$ is given in \eqref{Gamma} and $\Ha^i(t, x,u) = H^i(t, x,u,\Gamma^i(x,u))$.
\end{remark}

There are boundary conditions satisfied by the costate variable established by
the maximum principle, $p^i(T)=B^i_x(y(T))$. This, and the expression obtained from the
maximization of the Hamiltonian function provide a complete set of
final conditions for the MPNE system (\ref{PDE_MPNE}) given by:
\begin{equation}\label{curva}
L^i_{u^i}(T,x,u)+ F^{\top}_{u^i}(T,x,u)B^i_x(x)=0,\quad i=1,\ldots,N.
\end{equation}

In Rinc\'on-Zapatero et al. (1998) we show that under suitable
hypotheses about the Hamiltonians, a ${\mathcal C}^1$ solution
of (\ref{PDE_MPNE}) satisfying \eqref{curva}, becomes a MPNE of the differential
game. That is, (\ref{PDE_MPNE}) gives not only a set of necessary
conditions but also sufficient for optimality. The sufficiency
result can be summarized as follows:
\begin{theorem}
\hspace{.1cm} Let $\,u^*\in {\mathcal U}\,$ be a global
$\,{\mathcal C}^1\,$ solution of \eqref{PDE_MPNE}, \eqref{curva} interior
to the control region and satisfying $\,\det{(F_{u^i}(t,x,u^*))}\neq 0\,$ for all $(t,x)\in [0,T]\times \R^n$, for all
$i=1,\ldots, N$. Suppose further that for every $(t,x)\,$ and for
all $u^i\in {\mathcal U}^i$
\begin{equation}\label{aster}
H^i(t,x,\phii, \Gamma ^i (t,x,u^*))\le H^i(t,x,u^*, \Gamma
^i (t,x,u^*)),
\end{equation}
for all $i=1,\ldots,N$. Then $\, V^i_ x(t,x)=\Gamma ^i(t,x,u^*(t,x))\,$ and $\,u\,$ is a MPNE of the game
\eqref{Sisa}--\eqref{Fun}.
\end{theorem}
\begin{remark}
{\rm Theorem 1 states that the costate variables of the players
coincide with the gradient of the value function respect to $x$,
$\, p^i(t)=V^i_ x(t,x)$.}
\end{remark}

As a control problem is a particular case of a game with only one player, we can deduce readily the associated PDE system for an optimal solution $u$ of a control problem with the same characteristics than the game (interior solutions, equal dimension of state and control variables).
Remember that in the control problem OC=$(\ell,\rho,f,b)$, $\ell$ denotes the integrand, $f$ the dynamics, $\rho$ the discount factor and $b$ the bequest function; let
\begin{equation}\label{OChamiltonian}
h(t,x,u,p) = \ell(t,x,u) + f^{\top} (t,x,u)p
\end{equation}
be the Hamiltonian. Since the optimal control is interior, $h_u=0$, and this condition serves to define
\begin{equation}\label{OCcoestate}
\gamma(t,x,u) = -f_u^{\top}\ell_u(t,x,u)
\end{equation}
the vector of coestate variables. Note that $p(t)$ equals $\gamma(t,x(t), u(t,x(t))$, where $x(t)$ denotes the trajectory obtained by using the feedback control $u$. Also, let $\rho\ge 0$ be the discount factor. The PDE system \eqref{PDE_MPNE} becomes
\begin{equation}\label{PDE_OC}
 -\rho\gamma + \gamma_{u}(u_{t} + u_x f ) + h_x + \gamma _x f = 0.
\end{equation}
The extra complexity of system \eqref{PDE_MPNE} with respect to \eqref{PDE_OC} comes from the terms $ \sum_{j\neq i} u^{j\top}_{x}  H^i_{u^j}$ and $\sum_{j\neq i}^N \Gamma^i_{u^j}(u^j_{t} + u^j_x f )$, which are not present in \eqref{PDE_OC}. They are a consequence of the game interaction.

\subsection{A necessary condition for equivalence}\label{Sect_necc}
Recall that in our problem, given the game, we want to determine, if possible, a control problem with a similar structure where the profile of strategies $u=(u^1,\ldots,u^N)$ is the solution. By a ``similar structure", we mean that the dimension of the state variable is the same as the one of the game, that is, $n$, and that the number of controls is $N\times n$, which the total number of strategies of the game. However, we are free to choose both the dynamics and payoff functional, and even the discount factor, but within the class of functions that respect the game model. We will be more explicit about this in what follows. hus, we face an hypothetical control problem with $N\times n$ control variables, but only $n$ state variables. This does not fit into the framework where we have found \eqref{PDE_OC}. However, admitting, as we do, that an optimal solution exists, still we can find a PDE for the optimal control. 
Note that from the first order necessary condition for the Hamiltonian $h$
$$
\ell_{u^i} + f_{u^i}^{\top} p = 0
$$
for $u^i = (u^1_1,\ldots, u^i_ n)$, for $i=1,\ldots,N$, where $p$ is the vector of coestate variables. We have then $N$ linear systems for $p$, being consistent all of them, as we are supposing that optimal controls exist. So
$$
p = - f_{u^1}^{-\top} \ell_{u^1} =\cdots = - f_{u^N}^{-\top} \ell_{u^N},
$$
or, in terms of the vector functions $\gamma^i$ defined above,
\begin{equation}\label{Theta}
\gamma^1 (t,x,u) = \cdots = \gamma^N(t,x,u)
\end{equation}
at the optimal $u=u^*$. The next result shows conditions when \eqref{Theta} defines locally $N-1$ functions which depend smoothly on $x$ and on the controls of on of the players. Without loss of generality we choose $u^1$. In fact, along the text we will suppose that the aforementioned functions are globally defined.
\begin{lemma}\label{lemma:Theta}
Suppose that the $n\times N$ vector $u=u^*$ is a solution of the optimal control problem described above with $n$ state variables. Then, if the $n(N-1)\times n(N-1)$ determinant
\[
\left|
\begin{array}{c}
  (\gamma^1-\gamma^2)_{u^2,\ldots,u^N}   \\[1ex]
  \vdots  \\
   (\gamma^1-\gamma^N)_{u^2,\ldots,u^N}
\end{array}
\right|
\]
is not zero, then \eqref{Theta} defines a smooth solution
$$
u^j = \Theta^j(t,x,u^1),
$$
where $\Theta^j:[0,T]\times X\times U^1 \rightarrow U^j$, for $j=1,\ldots,N$, $\Theta^j=(\theta^1\equiv \Id, \theta^2,\ldots,\theta^n)$.
% and where $\Theta$ be the matrix $(\theta^i_j)$.
\end{lemma}
\begin{proof}
Without loss of generality, write \eqref{Theta} as $\gamma^1-\gamma^j=0$, for $j=2,\ldots,N$. This system of $n(N-1)$ equations admits solution by hypotheses. Let $(\gamma^1-\gamma^j)_{u^2,\ldots,u^N}$ be the $n\times n(N-1)$ matrix of derivatives of $\gamma^1-\gamma^j$ with respect to $u^2,\ldots,u^N$. The condition on the determinant of the lemma assures that the Implicit Function Theorem applies, and hence the system of equations define $u^2,\ldots,u^N$ as smooth functions of $t$, $x$ and $u^1$.
\end{proof}

This lemma implies that the system \eqref{PDE_MPNE} is overdetermined for $u^1$ ($n\times N$ PDEs for only $n$ unknowns $u^1$), which of course, impose certain limitations in the functions and elements that define the game. This is better observed in a class of scalar games with finite horizon and no explicit dependence with respect to the state variable, which is going to be analyzed in the following section.

\subsection{Scalar game with no explicit dependence on time nor on the state variable, with finite horizon and null discount}
Let a game DG as in \eqref{DG} with $n=1$, that is $X\subseteq \R$. Assume that $\Gamma_u=(\Gamma^i_{u^j})$ is invertible. As we are supposing that $n=1$, the system \eqref{PDE_MPNE} can be written
$$
\Gamma_u\left(u_t + Au_x\right) = \diag(r_1,\ldots,r_N) \Gamma - \Ha_x - F\Gamma_x.
$$
where $A = (FI_{N} 
+ \Gamma_u^{-1}H_u)$, $I_N$ the identity matrix of order $N$ and
$H_u=(H^i_{u^j})$. %Note that this matrix has null main diagonal, since $H^i_{u^i}=0$ from the Maximum Principle.
Here, $\diag(r_1,\ldots,r_N)$ is a diagonal matrix with the shown elements in the diagonal.
Suppose that $L^i_x=F_x=0$ and $r_i=0$, for all $i=1,\ldots,N$.

We look for a control problem OC as in \eqref{OC}, OC=$(\ell, \rho=0, f,  b)$, with the same characteristics, that is, with $\ell_x=f_x=0$. In this case, the vector $\Theta$ given in Lemma \ref{Theta} is independent from $t$ and $x$.

With null discount, $r_i=0$, and no explicit dependence on the state, $\Ha^i_x=0$ and $\Gamma^i_x=0$, for all $i=1,\ldots,N$, the system becomes homogenous
\begin{equation}\label{PDE_A}
u_t + Au_x = 0.
\end{equation}
The following result is well known in the PDE literature. We include it here, joint with the proof, for convenience of the reader. Translated to our framework, the lemma establishes a useful necessary condition for a MPNE of a game as described above to be capable of rationalisation by a control problem. We show in the main text its implications in an oligopoly game with additive externalities. 
\begin{lemma}\label{lemma:curvesolution}
Suppose that the system \eqref{PDE_A} admits a solution $u=(u^1,\ldots,u^N)$, of the form $u=\Theta(u^1)$ for some smooth function $\Theta=(\theta_1,\theta_2,\ldots,\theta_N)$, with $\theta_1=\Id$ and such that $u^1_x\neq 0$. Then, $$\Theta'=(1,\theta'_2,\ldots,\theta'_N)^{\top}$$ is a right eigenvector of $A$ and $u^1$ satisfies the quasilinear scalar PDE $$u^1_t+ \lambda(\Theta(u^1)) u^1_x = 0,$$ where $\lambda $ is an eigenvalue of matrix $A$.
\end{lemma}
\begin{proof}
Note that
$u_t=u^1_t\Theta'$, and $u_x=u^1_x\Theta'$, thus plugging these identities into \eqref{PDE_A}, we get
$$
\Theta'(u^1) u^1_t + A \Theta'(u^1)u^1_x= 0,
$$
and since $u^1_x\neq 0$, we see that $\Theta'$ is one of the eigenvectors of $A$ with eigenvalue $-u^1_t/u^1_x$. Thus, $u^1$ must satisfy the PDE $u^1_t + \lambda(\Theta(u^1))u^1_x = 0$, where $\lambda$ is one of the the eigenvalues of $A$.
\end{proof}

Along the same lines that the above result, the following one establishes a necessary condition for the coestate variables of the equivalent control problem.

\begin{lemma}\label{lemma:lefteigenvector}
Assume that $(\ell, \rho=0, f, b)$ rationalize the MPNE $u=(u^1,\ldots,u^N)$. Then $(\gamma^i_1,\ldots,\gamma^i_N)$ is a left eigenvector of matrix $A$ with associated eigenvalue $f$, for each $i=1,\ldots,N$.
\end{lemma}
\begin{proof}
By the hypotheses in the lemma, the MPNE satisfies both systems $u_t + A u_x = 0$ and $\gamma^i_u u_t + f \gamma^i_u u_x=0$, for $i=1,\ldots,N$. Multiplying the former system by $\gamma^i_u$, the latter system results if and only if the relation
$$
\gamma^i_u A = f\gamma^i_u
$$
holds. But this is the definition of left eigenvalue of $A$ with eigenvector $f$.
\end{proof}

\section{Asymmetric oligopoly}\label{AppB}
Asymmetry makes the problem difficult to handle.
For this reason, we consider ourselves to an asymmetric duopoly, $N=2$, and to infinite horizon; let us denote $u=u^1$ and $v=u^2$, to simplify notation. 
Suppose that the heterogeneity of the players is on self consumption, but they value the externality the same way. The preferences of the players are given by
$$
L^1(u,v)=(1-\alpha_1)^{-1}u^{1-\alpha_1}v^{1-\beta};\qquad L^2(u,v)=(1-\alpha_2)^{-1}v^{1-\alpha_2}u^{1-\beta}.
$$

Suppressing the time dependence term of the PDEs \eqref{PDE} for the MPNE, since we are in the infinite horizon game, the system becomes, assuming that $\alpha_1\alpha_2-(1-\beta)^2\neq 0$ and letting $\epsilon = (\alpha_1\alpha_2-(1-\beta)^2)^{-1}$

{\small
\begin{align}\label{Assym}
&\left(\begin{array}{cc}
-u-v +\epsilon \left(\frac{1-\beta}{1-\alpha_2}\right)\,\left((1-\alpha_2)u-(1-\beta)\,v\right) &
\epsilon \frac{\alpha_2}{1-\alpha_1}\frac  u{v}((1-\alpha_1)v-(1-\beta)u)
\\[2ex]
\epsilon \frac{\alpha_1}{1-\alpha_2}\frac  {v}u((1-\alpha_2)u-(1-\beta)\,v)
 &-u-v +\epsilon \left(\frac{1-\beta}{1-\alpha_1}\right)\,\left((1-\alpha_1)\,v-(1-\beta)u\right)\end{array}\right)
\left(\begin{array}{c} u_x\\[1ex]
v_x\end{array}\right)\\[1ex]
&\qquad\qquad = \left(\begin{array}{c} -{\epsilon\,\left(\alpha_2\,r_1+(1-\beta)\,r_2\right)u}\\[1ex] -\epsilon\,\left(\alpha_1\,r_2+(1-\beta)\,r_1\right)v\end{array}\right).\nonumber
\end{align}
}
Lemma \ref{lemma:Theta} in Section \ref{Sect_necc} in Appendix \ref{AppA} proves that a necessary condition for a MPNE to be the solution of an equivalent optimal control problem, is that there is 
a relation $v=\theta(t,x,u)$; in the game we are analyzing, where there is no explicit $t$ and $x$ dependence, the relation is simply $v=\theta(u)$ for suitable $\theta$. Of course, when the game is symmetric, as in the previous section, $\theta$ is the identity, but under asymmetry, $\theta$ is unknown.\footnote{However, we prove in Appendix \ref{AppA} that, if there is no discount factor in the game (and thus, the horizon is finite), $\theta$ can be nicely characterized in simple algebraic terms, see Lemma \ref{lemma:curvesolution}. This will be illustrated in Appendix \ref{AppC}, where we deal with the game with additive externalities. Here we continue working with Cobb--Douglas, multiplicative preferences.}
Plugging $v=\theta(u)$ into the above system, and given that $v_x=\theta' u_x$, the system becomes a pair of differential equations for only one unknown, $u(x)$, but where $\theta(u)$ is also unknown. The overdetermination can be used to find a differential equation for $\theta(u)$, 
by eliminating $u_x$. We simply take the ratio of both equations to get
{\small
\begin{align*}
&\frac{-(u+\theta) + \epsilon\left(\frac{1-\beta}{1-\alpha_2}\right)((1-\alpha_2)u-(1-\beta)\theta) + \epsilon \theta'\frac{u}{\theta}\left(\frac{\alpha_2}{1-\alpha_1}\right)((1-\alpha_1)\theta-(1-\beta )u)}{\epsilon \frac{\theta}{u}\left(\frac{\alpha_1}{1-\alpha_2}\right)((1-\alpha_2)u-(1-\beta)\theta) + \theta'\left(-(u+\theta)+\epsilon \left(\frac{1-\beta}{1-\alpha_1}\right)((1-\alpha_1)\theta-(1-\beta) u)\right)}\\
&\quad=\frac{u}{\theta}\left(\frac{\alpha_2 r_1+(1-\beta) r_2}{\alpha_1 r_2+(1-\beta) r_1}\right).
\end{align*}
}
This is an homogenous differential equation with an obvious solution candidate, $\theta(u)=\delta u$; for it to be sound, it is needed $\delta>0$.\footnote{The equation for $\delta$ is linear
$$
\frac{-(1+\delta) + \epsilon \left(\frac{1-\beta}{1-\alpha_2}\right)((1-\alpha_2)-(1-\beta)\delta) + \epsilon \delta\left(\frac{\alpha_2}{1-\alpha_1}\right)((1-\alpha_1)\delta-(1-\beta ))}{\epsilon \cancel{\delta}\left(\frac{\alpha_1}{1-\alpha_2}\right)((1-\alpha_2)-(1-\beta)\delta) - \cancel{\delta}(1+\delta)+ \epsilon \cancel{\delta}\left(\left(\frac{1-\beta}{1-\alpha_1}\right)((1-\alpha_1)\delta-(1-\beta) )\right)}=\frac{1}{ \cancel{\delta}}\left(\frac{\alpha_2 r_1+(1-\beta) r_2}{\alpha_1 r_2+(1-\beta) r_1}\right).
$$
}
Assuming values of the parameters such that a solution of this type exists, plugging it into \eqref{Assym}, we have that an optimal control problem $(\ell,f,\rho)$ rationalizes $u$ if, 
$
f(u) = \xi u
$
with\footnote{This comes from reading the first equation in the system \eqref{Assym}, taking the coefficient of $u_x$, and using \eqref{PDE_OC}. We follow the same steps than in the symmetric case studied in the previous section.}
$$
\xi = -(1+\delta) +\epsilon  \left(\frac{1-\beta}{1-\alpha_2}\right)((1-\alpha_2)-(1-\beta)\delta) + \epsilon \delta\left(\frac{\alpha_2}{1-\alpha_1}\right)((1-\alpha_1)\delta-(1-\beta ))
$$
and 
$$
\rho\frac{\gamma(u)}{\gamma'(u)} = -\epsilon \left(\alpha_2\,r_1+(1-\beta)\,r_2\right)u,
$$
where $\gamma (u) = -\frac{\ell'(u)}{f'(u)} = -\frac{\ell'(u)}{\xi}$. In this way, we get that the MPNE $u$ satisfies 
$$
f(u) u_x = \rho\frac{\gamma(u)}{\gamma'(u)},
$$
which is the ODE for the Markov control.

Clearly, $\gamma (u) =C u^{-\rho/(\alpha_2r_1+(1-\beta )r_2)}$ and hence 
$$
\ell(u) =-\lambda  C u^{1-\rho/(\alpha_2 r_1+(1-\beta )r_2)},
$$
where we have denoted generically different constants of integration by $C$. It is possible to choose $C$ with a suitable sign and $\rho>0$ with a suitable size to get an increasing strictly concave utility function.

\section{Additive externalities}\label{AppC}
In this section we assume that the agents' externality (if any) enters additively. Let
$$L^i(u^i,u_{-i}) = L_{io}(u^i) + \sum_{j\neq i} L_{ir}(u^j),$$
where $L_{io}$ (utility from own consumption, ``{\it o}") and $ L_{ir}$ (externality caused for the rest of players, ``{\it r}") are suitable smooth functions. Here, player $i$ values her own consumption with utility $ L_{io}$ and the consumption of the rest of the players additively, with the same utility $ L_{ir}$ for each player. We let $F(u^1,\ldots, u^N)=-\sum_{j=1}^N u^j$.

Following Remark \ref{PDEconservative} in Appendix \ref{AppA}, the costate function of player $i$ is
$
\Gamma^i(x,u^i)=L_{io}'(u^i)
$.
Thus, the matrix $\Gamma_u=(\Gamma^i_{u^j})$ is diagonal, with diagonal elements
$
L''_{1o}(u^1),\ldots,L''_{No}(u^N)
$.
The Hamiltonian is
$$
\Ha^i(x,(u^i|u_{-i}))=L_{io}(u^i) + \sum_{j\neq i} L_{ir}(u_j)-L_{io}'(u^i)\left(\sum_{j=1}^Nu^j\right),
$$
The matrix $A=FI_{N\times N} + \Gamma_u^{-1}H_u$ is then
$$
\left(\begin{array}{cccc} F(u) & 
E_{1}(u^1)-E_{12}(u^1,u^2)& \ldots &E_{1}(u^1)-E_{1N}(u^1,u^N) \\
E_{2}(u^2)-E_{21}(u^2,u^1) &F(u)& \ldots 
&E_{2}(u^2)-E_{2N}(u^2,u^N) \\
\vdots & \vdots & \ddots &\vdots\\
E_{N}(u^N)-E_{N1}(u^N,u^1)&E_{N}(u^N)-E_{N2}(u^N,u^2) & \ldots & F(u) \end{array}\right)
$$
\normalsize
where $E_{i}(u^i)=-\frac{L_{io}'(u^i)}{L_{io}''(u^i)}$ is the risk seeking index of Arrow-Pratt and $E_{ij}(u^i,u^j) = -\frac{L_{ir}'(u^j)}{L_{io}''(u^i)}$. Now it would be easy to obtain the PDE system for the MPNE.
Let us focus on the two player case, to show that in some instances there is no equivalence of the game with a control problem. We analyze the case of finite horizon and null discount factor. We will state first some auxiliary results based on Appendix \ref{AppA}.

%Based on the terminal conditions, note that in the finite horizon game and two players, $u^2=\theta(u^1)$ is given by
%$$
%L_1'(u^1) = S_1'(x),\qquad L_2'(\theta) = S_2'(x).
%$$
%Thus, $\theta$ is given by $\theta = (L_2')^{-1}\circ S_2'\circ (S_1')^{-1}\circ L_1'$, whenever all inverse functions involved are well defined.
%
%An interesting case happens when both bequest functions are linear, so that $S_1(x)=c_1x$ and $S_2(x)=c_2x$, say, with $c_i>0$, $i=1,2$. In this case, the equation above leads to $L'_1(u^1) = c_1$ and $L_2'(\theta(u^1)) = c_2$, which means that $\theta (\widehat{u}^1) = (L_2')^{-1}(c_2)$, where $\widehat u^1=(L_1')^{-1}(c_1)$.
%

\begin{lemma}\label{lemma:eigen}
Let the two person game with additive externalities and null discount described above. If 
$$
(E_{1}-E_{12})(E_{2}-E_{21})>0,
$$
then the matrix $A=FI_{2\times 2} + \Gamma_u^{-1}H_u$ admits two distinct real eigenvalues
\begin{align*}
\lambda &= F + \sqrt{(E_{1}-E_{12})(E_{2}-E_{21})},\\
\mu &= F - \sqrt{(E_{1}-E_{12})(E_{2}-E_{21})}
\end{align*}
with associated eigenspaces $S(\lambda)$ and $S(\mu)$ generated by
$$
s_{\lambda}(u^1,u^2) = \left(1,\frac{\sqrt{|E_{2}-E_{21}|}}{\sqrt{|E_{1}-E_{12}|}}\right)
$$
and
$$
s_{\mu}(u^1,u^2) =\left(1,-\frac{\sqrt{|E_{2}-E_{21}|}}{\sqrt{|E_{1}-E_{12}|}}\right),
$$
respectively. When 
$$
(E_{1}-E_{12})(E_{2}-E_{21})<0,
$$
the matrix $A$ has no real eigenvalues.
\end{lemma}
\begin{proof}
Remember that a non-null vector $(x,y)^{\top}$ is an eigenvector of
$$A= 
\left(
\begin{array}{cc}
 F &  E_{1}-E_{12}   \\
  E_{2}-E_{21}    &  F\end{array}
\right)
$$ with eigenvalue $\lambda$ iff $\det(A-\lambda I_{2\times 2})=0$ and $(A-\lambda I_{2\times 2})(x,y)^{\top}=(0,0)^{\top}$. Noting that 
$$
\det(A-\lambda I_{2\times 2}) = (F-\lambda)^2 - (E_{1}-E_{12} )(  E_{2}-E_{21} ),
$$
we obtain $\lambda$ and $\mu$ under the premise of the lemma. It is straightforward to find the eigenspaces for $\lambda$ and $\mu$. For instance
$
(A-\lambda I_{2\times 2})(x,y)^{\top}=(0,0)^{\top} $ iff
$$
 \left(
\begin{array}{cc}
 -\sqrt{(E_{1}-E_{12})(E_{2}-E_{21} )} &  E_{1}-E_{12}   \\
  E_{2}-E_{21}    &  -\sqrt{(E_{1}-E_{12})(E_{2}-E_{21} )} F\end{array}
\right) \left(
\begin{array}{c} x\\ y \end{array} 
\right) = \left(
\begin{array}{c} 0\\ 0 \end{array} 
\right),
$$
from which we get the expression for the first eigenvector in the lemma. The case for $\mu$ is similar.
\end{proof}
Now we establish necessary conditions that an oligopoly as described in this paper, is equivalent to a monopoly problem within the same class (nonrenewable resource, no discount factor, finite horizon).
\begin{proposition}\label{Prop:thetaODE}
Let the two person game with additive externality and null discount described above and let $u=(u^1,u^2)$ be a MPNE. A necessary condition for $u$ to be rationalized by a monopolistic nonrenewable resource model is 
\begin{equation}\label{Positive}
(E_{1}-E_{12})(E_{2}-E_{21})\ge 0.
\end{equation}
When the above inequality is strict, the pair $(u^1,u^2)$ is linked by $u^2=\theta(u^1)$, where $\theta$ satisfies one of the ODEs
\begin{equation}\label{ODE+}
\theta' (u^1) = \frac{\sqrt{|E_{2}(u^1)-E_{21}(\theta)|}}{\sqrt{|E_{1}(u^1)-E_{12}(\theta)|}}
\end{equation}
or
\begin{equation}\label{ODE-}
\theta' (u^1) = -\frac{\sqrt{|E_{2}(u^1)-E_{21}(\theta)|}}{\sqrt{|E_{1}(u^1)-E_{12}(\theta)|}}.
\end{equation}
Moreover, the dynamics of the control problem $f$ is one of the two possibilities (the positive or the negative one below)
\begin{equation}\label{f}
f = F \pm \sqrt{(E_{1}-E_{12})(E_{2}-E_{21})}.
\end{equation}
Also, the bequest functions $B^1$ and $B^2$ are linked as follows:
\begin{equation}\label{S1S2}
((L_2')^{-1}\circ B^2_x)(x) = (\theta\circ (L_1')^{-1}\circ B^1_x)(x),\quad x>0.
\end{equation}
On the other hand, if 
\begin{equation}\label{Negative}
(E_{1}-E_{12})(E_{2}-E_{21})<0,
\end{equation}
then no monopolistic nonrenewable resource model may rationalize the MPNE.
\end{proposition}
\begin{proof}
Condition \eqref{Positive} implies the existence of eigenvectors of the matrix $A$ by Lemma \ref{lemma:eigen} and then Lemma \ref{lemma:curvesolution} in  Appendix \ref{AppA} implies that if rationalization is possible, then $u^2=\theta(u^1)$, with $(1,\theta'(u^1))$ being one of the eigenvectors of $A$, which must be either $s_{\lambda}$ or $s_{\mu}$ defined in Lemma \ref{lemma:eigen}. These two possibilities lead to one of the ODEs for $\theta'$ displayed in the theorem. Now, by Lemma \ref{lemma:lefteigenvector} in Appendix \ref{AppA}, $f$ must be one of the eigenvalues of $A$, $\lambda$ or $\mu$.\footnote{Observe that they are obviously different from $F$ under the hypotheses of the theorem.}
%Also, both $\gamma^1_u$ and $\gamma^2_u$ are left eigenvalues with the same eigenvector

Regarding \eqref{S1S2}, when $x>0$, the final conditions satisfied by the pair $(u^1,u^2)$ at time $T$ are given by $L_i'(u^i) = B^i_x(x)$, $i=1,2$. Plugging $u^2=\theta (u^1)$ and eliminating $u^1$, we obtain the identity shown. 

Finally, \eqref{Negative} implies that no suitable $\theta$ may exist satisfying $u^2=\theta(u^1)$, as the matrix $A$ has no eigenvectors, thus the necessary condition for rationalizing the MPNE established in Lemma \ref{lemma:curvesolution} does not hold.
\end{proof}

The previous result allows to identify preferences of the players that cannot be replicated in a single--agent decision problem. These preferences must show some degree of asymmetry, in the sense that at least one of the players is jealous of the achievements of the other player (or, in another interpretation, that consumption of the other player is a negative externality).

\begin{theorem}\label{th:NO}
Let the two person game with additive externality and null discount factor described above and let $u=(u^1,u^2)$ be a MPNE. When the externality affecting to one of the oligopolist is negative, there are specifications of the preferences such that the game cannot be rationalized as a monopolistic nonrenewable resource model.
\end{theorem}
\begin{proof}
Let $L_{io}(u^i)=-e^{-\alpha_i u^i}$, with $\alpha _ i>0$, for $i=1,2$, $L_{1r}(u^2)=0$ and let $L_{2r}(u^1)= -d u^1$, where $d  \alpha_2>1$ and $d>0$ (hence, player 2 sees player 1 consumption as a negative externality). Note that $E_{i}(u^i) = \frac 1{\alpha_i}$, $E_{12}(u^2)=0$ and $E_{21}(u^1)=d$. Then 
$$
(E_{1}-E_{12})(E_{2}-E_{21}) = \frac 1{\alpha_1} \left(\frac 1{\alpha_2} - d\right) <0.
$$
Thus, \eqref{Negative} holds, and hence by Proposition \ref{Prop:thetaODE}, the MPNE cannot be rationalized as the solution of a monopolistic game.
\end{proof}

%\begin{remark} QUIZ\'A NO PONERLO
%
%With no externality, matrix $A$ has two different eigenvalues thus, in principle, there is no impossibility to represent the duopoly as a monopoly. The eigenvalues are
%$$
%f(u^1,u^2)=-(u^1+u^2) \pm \sqrt{E_1(u^1)E_2(u^2)}.
%$$
%For instance, in the HARA case if $L_i(u^i)=(1-\alpha_i)^{-1}(u^i)^{1-\alpha^i}$, then
%$$
%f(u^1,u^2)=-(u^1+u^2) \pm \sqrt{\frac{u^1u^2}{\alpha^1\alpha^2}}.
%$$
%The eigenvectors are $\left(1,\pm\sqrt{\frac{E_2(u^2)}{E_1(u^1)}}\right)$. Hence, the ODEs for $\theta$ are
%$$
%\theta'(u^1) = \pm\sqrt{\frac{E_2(\theta(u^1))}{E_1(u^1)}}.
%$$
%In the HARA case (let use the notation $u^1=u$)
%$$
%\theta'= \sqrt{\frac{\alpha_1\theta}{\alpha_2 u}}.
%$$
%This separable equation can be solved explicitly.
%$$
%v=\theta(u)=\left(\sqrt{\frac{\alpha_1}{\alpha_2 }u} + C\right)^2,
%$$
%with $C$ a constant. 
%
%
%\end{remark}

It is worth noting that in the problem that we are analyzing, the dynamics of the control problem should be given by $\lambda$ or $\mu$ defined in the above theorem, and not by the original $F$. It comes as a surprise that the game cannot be put as a control problem if we insist in maintaining the original dynamics, $F$. The only exception is when the game is a team problem, where both players have the same objective, since then $H^i_{u^j}=0$ and the matrix is diagonal, with diagonal $(F, F)$, and thus the only eigenvalue of $A$ is obviously $F$.

The following theorem is a positive result, in the sense that it establishes conditions making possible the rationalization of the MPNE in a duopoly with additive externalities.
\begin{theorem}
Let the two person game with additive externality and null discount factor described above and let $u=(u^1,u^2)$ be a MPNE such that \begin{enumerate}
\item Inequality
\eqref{Positive} holds;
\item 
ODE \eqref{ODE+} or ODE \eqref{ODE-} admit a feasible solution $\theta (u) \ge 0$;
\item Identity
\eqref{S1S2} holds;
\item
Letting $f(u)=\lambda(u,\theta(u))$ or $f(u)=\mu(u,\theta(u))$, $f$ is strictly monotone and concave or convex.
\end{enumerate}
Then, there is a strictly concave payoff function $\ell$ and a bequest function $b$ such that the monopolistic nonrenewable resource model $(\ell,\rho = 0,f,b)$ implements the MPNE.
\end{theorem}
\begin{proof} Without loss of generality, suppose that part (4) of the theorem holds for $f(u)=\lambda(u,\theta(u))$. Notice that $f$ is well defined due to the assumptions made and the previous results (if needed, choose the other eigenvalue). We have to choose $\ell$ and the bequest function $b$. From the terminal condition for $u^i$, we get $x=(B^i_x)^{-1}(L_i'(u^i))$, and there is no ambiguity here for $x$, since we assume that \eqref{S1S2} holds. Now, look the final condition for the control problem at $t=T$: it reads $\ell'(u) + f'(u) b'(x) = 0$. Let the function $\psi(\cdot ) = (B^1_x)^{-1}(L_1'(\cdot ))$. Note that the hypotheses set on the game imply that $\psi$ is  strictly monotone. Plug $x = \psi(u)$ into the equation to obtain for the derivative of $\ell$ the expression
$$
\ell'(u) = -f'(u) b'(\psi(u)).
$$
Integrate to get $\ell$. Now, to force strict concavity of $\ell$ (more than that: $\ell''<0$), we choose a suitable function $b$, depending on the signs of $f'$, $f''$ and $\psi'$. We start from $\ell'' = - \left(f'' b' + f'b'' \psi'\right)$ and study below all possible cases.
\begin{itemize}
\item $f'\psi'>0$, $f''\ge 0$: choose $b$ strictly increasing and convex;
\item $f'\psi'>0$, $f''\le 0$: choose $b$ strictly decreasing and convex;
\item $f'\psi'<0$, $f''\ge 0$: choose $b$ strictly increasing and concave;
\item $f'\psi'<0$, $f''\le 0$: choose $b$ strictly decreasing and convex;
\end{itemize}
The above covers all possibilities, making $\ell$ strictly concave. Continuing with the proof, notice that the PDE for the control problem becomes
\begin{equation}\label{PDEcontrol}
\gamma'(u) (u_t + \lambda(u,\theta(u)) u_x)= 0,
\end{equation}
where $\gamma(u)=-\ell'(u) / f'(u)$; note that $\gamma'(u)\neq 0$ since $\ell''<0$.
The PDE system for the MPNE collapse into $u_t + \lambda(u,\theta(u)) u_x=0$, by the way that $\lambda$ and $\theta'$ are chosen, as explained in the results above. Also, the final condition for the MPNE is fulfilled by hypothesis. Thus, the control problem solution $u$ gives rise to the MPNE $(u^1,u^2)=(u^1,\theta (u^1))$ of the game.

\end{proof}

\end{document}